\documentclass[a4paper,11pt]{article}
\usepackage[a4paper,margin=26mm]{geometry}
\usepackage[utf8]{inputenc}
\usepackage{cite} % sort references within the same \cite command               
\usepackage{enumerate} % fancy options for enumerate environment                
\usepackage{amsmath,amsfonts,amssymb,amsthm}
\usepackage{graphicx,epsfig,wrapfig}
\usepackage{url,hyperref}
\usepackage{color,xcolor}
\usepackage{microtype} % better line breaks and fewer overfull hboxes           
\usepackage{stmaryrd}

\usepackage{dsfont}
\usepackage[charter]{mathdesign}

\newcommand{\DD}{\mathbb{D}}
\newcommand{\RR}{\mathbb{R}}
\newcommand{\ZZ}{\mathbb{Z}}
\newcommand{\BB}{\mathbb{B}}
\newcommand{\SSS}{\mathbb{S}}

\newcommand{\poly}{\text{poly}}
\newcommand{\eps}{\varepsilon}
\def\OPT{\text{OPT}}

\def\DEF#1{\textbf{\emph{#1}}}

\DeclareMathOperator{\cost}{cost}
\DeclareMathOperator{\opt}{opt}
\DeclareMathOperator{\Center}{center}

\newtheorem{theorem}{Theorem}
\newtheorem{lemma}[theorem]{Lemma}
\newtheorem{proposition}[theorem]{Proposition}

\title{Hardness of Minimum Barrier Shrinkage and\\ Minimum Installation Path}

\author{Sergio Cabello\thanks{Faculty of Mathematics and Physics, University of Ljubljana,
  Slovenia, and Institute of Mathematics, Physics and Mechanics,
  Slovenia.  Email: \texttt{sergio.cabello}\texttt{@fmf.uni-lj.si}.}\and
\'Eric Colin de Verdi\`ere\thanks{LIGM, CNRS, Univ Gustave Eiffel, ESIEE
  Paris, F-77454  Marne-la-Vall\'ee, France.  Email: \texttt{eric.colindeverdiere}\texttt{@u-pem.fr}.}}

\begin{document}
\maketitle

\begin{abstract}
	In the \textsc{Minimum Installation Path} problem, we are given a graph $G$ with edge weights $w(\cdot)$ 
	and two vertices $s,t$ of $G$.
	We want to assign a non-negative power $p\colon V\rightarrow
        \RR_{\geq 0}$ to the vertices of $G$,
	so that the activated edges $\{ uv\in E(G) \mid p(u)+p(v)\geq w(uv)\}$ contain some $s$-$t$-path,
	and minimize the sum of assigned powers.
	In the \textsc{Minimum Barrier Shrinkage} problem, we are given, in
        the plane, a family of disks
	and two points $x$ and $y$. The task is to shrink the disks, 
	each one possibly by a different amount, so that we can draw an $x$-$y$ curve that
	is disjoint from the interior of the shrunken disks, and the sum of the decreases in the radii
	is minimized. 
	
	We show that the \textsc{Minimum Installation Path} and the
        \textsc{Minimum Barrier Shrinkage} problems (or, more precisely,
        the natural decision problems associated with them) are weakly NP-hard.

\end{abstract}

%%%%%%%%%%%%%%%%%%%%%%%%%%%%%%%%%%%%%%%%%%%%%%%%%%%%%%%%%%%%%%%%%%%%%%%%%%%%%%%%%%%%%%%%%%%%%%%%%%%%%%%%%%%%%%%%%%%%%%%%%%%%%%%%%%%%%%%%%%%%%%%%%%%%%%%%%%%%%%%%%%%%%%%%%%%%%%%%%%%
\section{Introduction}

Let $X$ be a subset of the plane, let $x$ and $y$ be points in $X$,
and let $\SSS$ be a family of shapes in the plane.
An \DEF{$x$-$y$ curve} is a curve in $\RR^2$ with endpoints $x$ and $y$.
We say that $\SSS$ \DEF{separates} $x$ and $y$ in $X$ if each $x$-$y$ curve contained
in $X$ intersects some shape from $\SSS$.
Let $D(c,r)$ denote the \emph{open} disk centered at $c$ with radius $r$.

In this work we show that the following two decision problems are weakly NP-hard.
This means that in our reduction we will use numbers that are exponentially large,
but have polynomial length when written in binary.

\begin{quote}
	\textsc{Minimum Barrier Shrinkage.}\\
	Input: a family $\{ D(c_i,r_i) \mid i=1,\dots,n \}$ of $n$ open
        disks; two points $x,y\in \RR^2$; a real number~$C$.\\
	Output: Whether there exist \DEF{shrinking values} $\delta_1,\dots,\delta_n\ge 0$ such that
		their \DEF{cost} $\sum_i \delta_i$ is at most~$C$ and
		the family of open disks $\{ D(c_i,r_i-\delta_i) \mid i=1,\dots n \}$ 
		does not separate $x$ and $y$ in $\RR^2$.
\end{quote}
\begin{quote}
	\textsc{Minimum Installation Path.}\\
	Input: a graph $G=(V,E)$ with positive edge weights $w\colon
        E\rightarrow \RR_{>0}$; two vertices $s$ and $t$ of $G$; a real
        number~$C$.\\
	Output: Whether there exists an assignment of \DEF{powers} $p\colon V\rightarrow \RR_{\geq 0}$ to the
		vertices such that its \DEF{cost}
		$\sum_{v\in V} p(v)$ is at most~$C$ and the \DEF{activated edges}
		$E(p)=\{ uv\in E\mid p(u)+p(v)\geq w(uv)\}$ contain an $s$-$t$-path.
\end{quote}

We next discuss the motivating and closest related work.

\paragraph{Minimum barrier shrinkage}
Kumar, Lai and A.~Arora~\cite{kumar2007barrier} introduced the following
\DEF{barrier resilience problem} in the plane.  The input is specified by a
domain $X\subseteq \RR^2$, a family $\DD$ of disks in~$\RR^2$, and two
points $x$ and $y$ in~$X$.  The task is to find an $x$-$y$ curve in $X$
that intersects as few disks of $\DD$ as possible, without counting
multiplicities.  An alternative statement is that we want to find a minimum
cardinality subfamily $\DD'\subseteq \DD$ such that $\DD\setminus \DD'$
does not separate $x$ and $y$ in $X$.  The intuition is that we have
sensors detecting movements from $x$ to $y$, and we want to know how many
sensors can suffer a total failure and still any agent moving from $x$ to
$y$ within $X$ is detected by some of the remaining sensors.

Kumar, Lai and A.~Arora~\cite{kumar2007barrier} showed that the problem can 
be solved in polynomial time when the domain $X$ is a vertical strip bounded between 
two vertical lines $\ell$ and $\ell'$, the point~$x$ lies above and the
point~$y$ lies below
all disks of $\DD$. Let us call this scenario the \DEF{rectangular scenario}.
The main insight is to consider the intersection graph $G$
defined by $\DD\cup \{ \ell, \ell'\}$ and to note that the solution is the maximum number
of $\ell$-$\ell'$ internally vertex-disjoint paths in $G$. Thus, the problem
can be solved in polynomial time by solving maximum flow problems.
The same argument works for any
family of shapes $\SSS$, not just disks, as far as each shape of $\SSS$ is connected.

Despite the claim in the preliminary version~\cite{Kumar:2005:BCW:1080829.1080859} 
of~\cite{kumar2007barrier}, we do not know whether the barrier resilience problem can be solved exactly
in polynomial time when the domain $X$ is all of~$\RR^2$. 
In fact, we know that when $X=\RR^2$ and the family $\DD$ of disks is replaced by 
some other family $\SSS$ of shapes, 
the problem is NP-hard~\cite{AltCGK17,KormanLSS13,TsengK11}.
The difference between the strip and the whole plane is that in the former case
we can use Menger's theorem to relate the number of $\ell$-$\ell'$ paths in the
intersection graph of $\SSS\cup \{ \ell, \ell'\}$ to the $\ell$-$\ell'$ vertex connectivity,
but no such statement applies to cycles that ``separate'' $x$ and $y$.
The computational complexity of the barrier problem in the plane 
for (unit) disks and (unit) squares is a challenging open problem,
and several approximation algorithms have been devised~\cite{BeregK09,ChanK12,KormanLSS13}.

Modeling the fact that sensors are less reliable further away from their placement,
Cabello et al.~\cite{cabello2018minimum} considered the problem of minimizing the total
shrinkage of the disks such that there is an $x$-$y$ curve disjoint from the interior
of the disks. This is precisely the problem \textsc{Minimum Barrier Shrinkage}. 
Cabello et al.~also provided an FPTAS for the rectangular scenario. 
The algorithm uses the connection to vertex-disjoint paths.

We believe that showing NP-hardness for the problem \textsc{Minimum Barrier Shrinkage} is
interesting because of the computational complexity of two closely related problems,
the barrier resilience problem for $X=\RR^2$ and 
the minimum barrier shrinkage problem in the rectangular scenario, are unknown.

\paragraph{Minimum installation path}

There is a rich literature on so-called Activation Network problems.  The
task is to assign a power $p(v)$ to each vertex~$v$ of an edge-weighted
graph $G=(V,E)$ so that the activated edges satisfy a certain connectivity
property, such as for example spanning the whole graph.  Whether an edge
$uv$ is activated depends only on $p(u)$ and~$p(v)$.  In the most general
scenario, one only assumes an oracle telling, given $p(u)$ and~$p(v)$,
whether $u$ and~$v$ are activated, together with a natural monotonicity
constraint: if some choice of $p(u)$ and~$p(v)$ activates $uv$, then
increasing the powers at $u$ and~$v$ still leaves $uv$ activated.  In many
cases, the following simplifying assumption is made: the possible powers at
the vertices are discretized as a finite set of values, denoted by~$D$ (the
\emph{domain}).  See the survey by Nutov~\cite{nutovactivation} for an
overview of the area.

In this context, Panigrahi~\cite[Section~4.1]{Panigrahi2011} considered the
\textsc{Minimum Activation Path} problem: the connectivity constraint is
that the activated edges must include a path between two fixed vertices $s$
and~$t$ of~$G$.  He provided an algorithm with running time
$O(\poly(n,|D|))$, where $n$ is the number of vertices of~$G$ and~$D$ is
the finite domain of values for the power assignments.

Compared to the problem studied by Panigrahi, our \textsc{Minimum
  Installation Path} has two differences.  First, power assignments are not
discretized and can be arbitrary nonnegative real numbers.  Second, whether
an edge $uv$ is activated is simply determined by whether
$p(u)+p(v)\ge w(uv)$.  In this article, we show that the \textsc{Minimum
  Installation Path} is weakly NP-hard.  We also provide a simple fully
polynomial time approximation scheme (FPTAS) relying on the algorithm by
Panigrahi~\cite{Panigrahi2011}.

Our weak NP-hardness result of the \textsc{Minimum Installation Path} problem
is consistent with the result of Panigrahi.
In our reduction, we use large \emph{integer} weights: they have a polynomial bit length,
but they are exponentially large.
Taking $D=\ZZ\cap [0,\max_{uv}w(uv)]$ in the algorithm of Panigrahi,
one only gets a pseudopolynomial time algorithm for such instances. 
This is consistent with a weakly NP-hardness proof.

In a similar vein, we remark that Alqahtani and Erlebach~\cite{alqahtani}
presented algorithms parameterized by the treewidth of the graph in the
case where the goal is to activate $k$ node-disjoint $st$-paths, or
node-disjoint paths between $k$ pairs of terminals.  See also Lando and
Nutov~\cite{lando-nutov} and Althaus et al.~\cite{althaus}.

\paragraph{Relation between the problems}
We are not aware of any polynomial-time reduction from one problem to the other.
Nevertheless, the NP-hardness proofs for both problems are very similar.
The underlying connection between both problems is the following classical 
property: 
in a planar graph $G=(V,E)$, a set of edges $F\subset E$ is a minimum $s$-$t$ cut 
if and only if, in the dual graph~$G^*$, the edges $\{e^*\mid e\in F\}$ 
form a shortest cycle separating the face $s^*$ from
the face $t^*$. 
This relation does not directly provide a reduction even
in the case of planar graphs, but does inspire the adaptation we make.
Actually, our hardness proof for \textsc{Minimum Barrier Shrinkage} reuses
components of the hardness proof for \textsc{Minimum Installation Path}; we
reformulate some special instances of \textsc{Minimum Barrier Shrinkage} in
terms of graphs and then remark that each reformulated instance is equivalent
to an instance of \textsc{Minimum Installation Path}.

\paragraph{Organization}

It seems more convenient to present the NP-hardness of \textsc{Minimum
  Installation Path} first.  We achieve this in
Section~\ref{sec:installation}, together with an FPTAS for this problem.
Then, in Section~\ref{sec:barrier}, we show that the \textsc{Minimum
  Barrier Shrinkage} problem is NP-hard.

%%%%%%%%%%%%%%%%%%%%%%%%%%%%%%%%%%%%%%%%%%%%%%%%%%%%%%%%%%%%%%%%%%%%%%%%%%%%%%%%%%%%%%%%%%%%%%%%%%%%%%%%%%%%%%%%%%%%%%%%%%%%%%%%%%%%%%%%%%%%%%%%%%%%%%%%%%%%%%%%%%%%%%%%%%%%%%%%%%%%
\section{Minimum installation path}
\label{sec:installation}
In this section we study the complexity of the \textsc{Minimum Installation
  Path} problem is NP-hard.  We first provide a simple FPTAS for this
problem, and then prove that it is weakly NP-hard.

\subsection{A simple FPTAS}
\label{sec:installation_fptas}

As a side note, we show that the main idea used by Cabello et
al.~\cite{cabello2018minimum} can be adapted to lead to a simple FPTAS for
\textsc{Minimum Installation Path}.  Let us consider an instance of that
problem.

\begin{lemma}\label{L:computelambda}
  In polynomial time, we can compute the smallest value~$\lambda$ such that
  setting $p(v)=\lambda$ for all vertices~$v$ of~$G$ activates at least one
  $st$-path.
\end{lemma}
\begin{proof}
  Whether one $st$-path is activated by the power assignment $p(v)=\lambda$
  (for each vertex~$v$) depends only on the set of activated edges.  So,
  for some edge~$uv$, the minimum value of~$\lambda$ activating at least
  one $st$-path is the minimum value of~$\lambda$ activating edge~$uv$.  In
  other words, the minimum value of $\lambda$ is necessarily of the form
  $w(uv)/2$ for some edge~$uv$.  So, for each edge $uv$, we determine
  whether putting power $w(uv)/2$ to all vertices activates an $st$-path,
  and return the smallest value that does so.
\end{proof}

\begin{lemma}\label{L:proplambda}
  Let \OPT{} be the optimum value of the \textsc{Minimum Installation Path}
  instance.  Then:
  \begin{enumerate}
  \item in an optimal solution, the power assigned to every vertex is at
    most $n\lambda$, where $n$ is the number of vertices of the input
    graph~$G$;
  \item $\lambda\le\OPT$.
  \end{enumerate}

\end{lemma}
\begin{proof}
  \begin{enumerate}
  \item $\OPT\le n\lambda$ because the definition of~$\lambda$ implies a
    feasible solution of cost $n\lambda$.  This implies~(1);
  \item $\lambda\le\OPT$ because otherwise, some $st$-path would be
    activated by some powers strictly smaller than~$\lambda$ at each
    vertex, contradicting the definition of~$\lambda$.\qedhere
  \end{enumerate}
\end{proof}

\begin{proposition}
  \textsc{Minimum Installation Path} admits an FPTAS.
\end{proposition}
\begin{proof}
  Let $\eps>0$ be given.  We first compute~$\lambda$ using
  Lemma~\ref{L:computelambda}.  Then we apply the algorithm by
  Panigrahi~\cite[Section~4.1]{Panigrahi2011} to the instance, in which the
  domain is defined by
  \[D=\{k\eps\lambda/n\ \mid\ k=0,\ldots,n^2/\eps\}.\]

  By Lemma~\ref{L:proplambda}(1), we obtain a feasible solution to the
  original problem, the cost of which is within an additive error of at
  most $\eps\lambda/n$ per vertex from~\OPT, hence with an additive error
  of at most $\eps\lambda$ overall.  By Lemma~\ref{L:proplambda}(2), this
  is at most $\eps\OPT$.  Clearly the running time is polynomial in $n$
  and~$1/\eps$.
\end{proof}

We can readily extend this argument to more general activation functions.
For example, assume that each edge $uv$ is activated if and only if
$\alpha(uv)p(u)+\beta(uv)p(v)\ge w(uv)$, for some positive constants
$\alpha(uv)$, $\beta(uv)$, and~$w(uv)$.  (Our setup corresponds to
$\alpha(uv)=\beta(uv)=1$.)  The same argument as above shows that this
extended version of \textsc{Minimum Installation Path} admits an FPTAS.

\subsection{Greedy solution in a path}

In the rest of Section~\ref{sec:installation}, we focus on proving
NP-hardness of \textsc{Minimum Installation Path}.  We first consider the
particular case of a path.

Consider a graph~$G$ and a path $\pi=v_0,\dots,v_n$ in $G$. 
We define greedily a
power assignment $p^*_\pi$ on the vertices of~$G$ to activate $\pi$, 
in a way that power is
pushed forward along $\pi$ as much as possible.
Formally, the \DEF{greedy power assignment along $\pi$} is
	\begin{equation}
		p^*_\pi(v) ~=~\begin{cases}
					0 &\text{if $v$ does not belong to $\pi$ or $v=v_0$,}\\
					0 &\text{if $v=v_i$, $i>0$ and $p^*_\pi(v_{i-1})\ge w(v_{i-1}v_i)$},\\
					w(v_{i-1}v_i)-p^*_\pi(v_{i-1}) &\text{if $v=v_i$, $i>0$ and $p^*_\pi(v_{i-1})< w(v_{i-1}v_i)$.}
				\end{cases}
	\label{eq:powers}
	\end{equation}
	
For a power assignment~$p$, let $\cost(p)$ denote the total cost of~$p$,
namely, the sum of the powers at the vertices. 
For path $\pi$, let $\opt(\pi)$ be the cost of the minimum cost 
power assignment that activates $\pi$. The following lemma tells that 
the greedy power assignment along $\pi$ has minimum cost to activate $\pi$.

\begin{lemma}
\label{le:path}
For each path $\pi$, $\cost(p^*_\pi)=\opt(\pi)$.
\end{lemma}
\begin{proof}
  It is clear that $p^*_\pi$ activates all the edges of~$\pi$.  Let $p$ be
  another power assignment activating all edges of~$\pi$.  We have to show that
  $\cost(p^*_\pi)\le \cost(p)$.
  
  We can assume that $p(v)=0$ at all vertices $v$ outside~$\pi$. 
  Otherwise, we change~$p$ to have this property. This
  reassignment of power would decrease the cost and would keep activating 
  the path~$\pi$.
  
  The strategy is to
  gradually transform~$p$ into~$p^*_\pi$ while keeping all edges of~$\pi$
  activated and without increasing the value of
  $\cost(p)$.   
  The property $\cost(p^*_\pi)\le \cost(p)$ is trivially correct if $p=p^*_\pi$. 
  So assume $p\ne p^*_\pi$ and let
  $i$ be the smallest integer such that $p(v_i)\ne p^*_\pi(v_i)$.  Because all
  edges are activated, and by construction of~$p^*_\pi$, we must have
  $p(v_i)>p^*_\pi(v_i)$.  Let $\Delta=p(v_i)-p^*_\pi(v_i)>0$.  There are two
  cases:
  \begin{itemize}
  \item Assume $i\le n-1$.  Update~$p$ by decreasing $p(v_i)$ by~$\Delta$
    and increasing $p(v_{i+1})$ by~$\Delta$.  Since each edge of~$\pi$ is activated
    by~$p^*_\pi$ and by~$p$ before this transformation, 
    each edge of~$\pi$ is still activated by the new~$p$.  Moreover, $\cost(p)$ 
	is unchanged.
  \item Assume $i=n$.  Update $p$ by decreasing $p(v_n)$ by~$\Delta$.
    Again, each edge of~$\pi$ is still activated.  The cost has decreased
    by~$\Delta$.
  \end{itemize}
  This transformation does not increase the value of $\cost(p)$. 
  Moreover, the new power assignment coincides
  with~$p^*_\pi$ on vertices $v_0,\ldots,v_i$.  Thus, after a finite number of
  steps, $p=p^*_\pi$.  This proves the lemma.
\end{proof}

For the path $\pi=v_0,\dots,v_n$, let $\varphi(\pi)=p^*_\pi(v_n)\ge0$.
That is, $\varphi(\pi)$ is the power assignment given by the greedy power assignment along $\pi$
to the final vertex. 
Since $p^*_\pi(v_n)$ depends on $p^*_\pi(v_{n-1})$, we have the following.

\begin{lemma}
\label{le:prolong}
	Let $\pi$ be the path $v_0,\dots,v_n$ and let $\pi'$
	be the path $v_0,\dots,v_n,u$. (Thus, $\pi'$ extends $\pi$ by an additional edge $v_nu$.)
	Then $\varphi(\pi')= \max\{ 0, w(v_n u)-\varphi(\pi) \}$ and
	$\opt(\pi')=\opt(\pi)+ \varphi(\pi')$.
\end{lemma}
\begin{proof}
  From the definition of the greedy power assignment along $\pi$ and $\pi'$,
	the power assignments $p^*_\pi$ and $p^*_{\pi'}$ differ only 
	at vertex $u$.  We have:
	\begin{align*}
		\varphi(\pi') = p^*_{\pi'}(u) ~&=~ \max\{ 0, w(v_n u)-p^*_{\pi'}(v_n) \}
				~=~ \max\{ 0, w(v_n u)-p^*_\pi(v_n) \}\\
				~&=~ \max\{ 0, w(v_n u)-\varphi(\pi)) \}.
	\end{align*}
  This proves the claim for~$\varphi(\pi')$.  
	Because of Lemma~\ref{le:path} for $\pi$ and $\pi'$ we also get
	\begin{align*}
		\opt(\pi') ~&=~ \cost(p^*_{\pi'}) ~=~ \cost(p^*_\pi) + p^*_{\pi'}(u) - p^*_\pi(u)\\
				  ~&=~ \opt(\pi)+ \varphi(\pi') -0. \qedhere
	\end{align*}
\end{proof}

A consequence of Lemma~\ref{le:path} is the following integrality property.
\begin{lemma}
  \label{le:integral}
  Assume that the weight function $w\colon E(G)\rightarrow \RR_{>0}$ takes
  only integer values, and that $C$ is also an integer.  Then, for any
  $\alpha\in[0,1)$, \textsc{Minimum Installation Path}$(G,w,s,t,C)$ has a
  positive answer if and only if \textsc{Minimum Installation
    Path}$(G,w,s,t,C+\alpha)$ has a positive answer.
\end{lemma}
\begin{proof}
  Assume that \textsc{Minimum Installation Path}$(G,w,s,t,C+\alpha)$ is has a
  positive answer.  Consider a power assignment $p$ corresponding to a
  feasible solution of minimum cost (at most~$C+\alpha$); let $\pi$ be an
  $s$-$t$ path activated by $p$.  Because of Lemma~\ref{le:path} we have
  $\cost(p)=\opt(\pi)=\cost(p^*_\pi)$.  From the inductive definition
  \eqref{eq:powers} of $p^*_\pi$, we see that $p^*_\pi$ assigns integral
  powers to all vertices, and thus $\cost(p^*_\pi)=\sum_v p^*_\pi(v)$ is an
  integer, which is at most~$C$.  So \textsc{Minimum Installation
    Path}$(G,w,s,t,C)$ has a positive answer.
\end{proof}

%%%%%%%%%%%%%%%%%%%%%%%%%%%%%%%%%%%%%%%%%%%%%%%%%%%%%%%%%%%%%%%%%%%%%%%%%%%%%%%%%%%%%%%%%%%%%%%%%%%%%%%%%%%%%%%%%%%%%%%%%%%%%%%%%%%%%%%%%%%%%%%%%%%%%%%%%%%%%%%%%%%%%%%%%%%%%%%%%%%%
\subsection{The reduction}
\label{sec:reduction}

Now we provide the reduction. The reduction is inspired by the reduction
used to show that the restricted shortest path problem is NP-hard;
this seems to be folklore and attributed to Megiddo 
by Garey and Johnson~\cite[Problem ND30]{GareyJ79}.
We use the notation $[n]=\{ 1,\dots, n\}$
and reduce from the following problem.

\begin{quote}
	\textsc{Subset Sum}\\
	Input: a sequence $a_1,\dots,a_n$ of positive integers and a positive integer $b$.\\
	Question: is there a set of indices $I\subseteq [n]$ such that $\sum_{i\in I} a_i=b$?
\end{quote}

The problem \textsc{Subset Sum} is one of the standard weakly NP-hard
problems that can be solved in pseudopolynomial time via 
dynamic programming~\cite[Section 4.2]{GareyJ79}.
In particular, when the numbers $a_i$ are bounded by a polynomial in $n$,
the problem can be solved in polynomial time.

Set $L$ to be an integer strictly larger than $2\sum_{i\in [n]} a_i$. 
Then, for each $I\subseteq [n]$ we have $2\sum_{i\in I} a_i< L$.

We construct a graph $G=G(a_1,\dots,a_n,b)$ as follows (see Figure~\ref{fig:installation1}).
$G$ will include vertices $s$, $t$, $u_1,\dots,u_n$.
Let us use the notation $u_0=s$.
For each $i\in [n]$, we put between $u_{i-1}$ and $u_i$ two paths, each of
length two, 
one path with weights $L+2a_i$ and $L+2a_i$,
and the other path with weights $L+ a_i$ and $L+3 a_i$,
as we go from $u_{i-1}$ to $u_i$.
Finally, we put the edge $u_n t$ with weight $2b$.
This finishes the construction of $G$.

\begin{figure}
	\includegraphics[width=\textwidth,page=1]{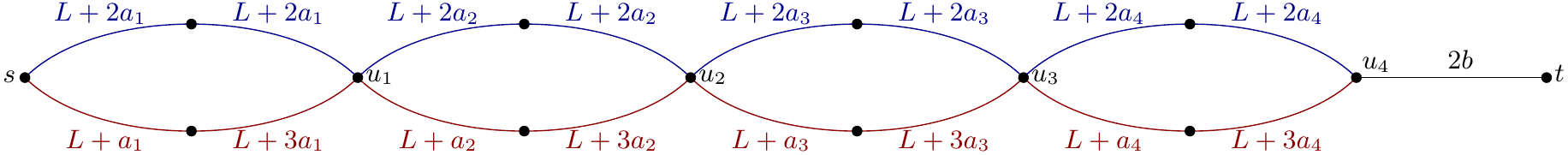}
	\caption{The graph $G$ when $n=4$.}
	\label{fig:installation1}
\end{figure}

\begin{lemma}
\label{le:feasible}
	There exists a path $\pi$ from $s$ to $u_n$ in $G$
	with $\opt(\pi)=c$ and $\varphi(\pi)=r$ 
	if and only if there exists $I\subseteq [n]$ such that
	\[
		r = 2\sum_{i\in I} a_i 
		~~~~~\text{and}~~~~~
		c = nL + 2\sum_{i\in [n]} a_i + \sum_{i\in I} a_i.
	\]
\end{lemma}
\begin{proof}
  Consider the two paths, each of length two, connecting $u_{i-1}$ to
  $u_i$.  The \emph{upper choice at $i$} is the path with weights $L+2a_i$;
  similarly, the \emph{lower choice at $i$} is the path with weights
  $L+a_i$ and $L+ 3 a_i$. See Figure~\ref{fig:installation2}.

	\begin{figure}
		\includegraphics[width=\textwidth,page=2]{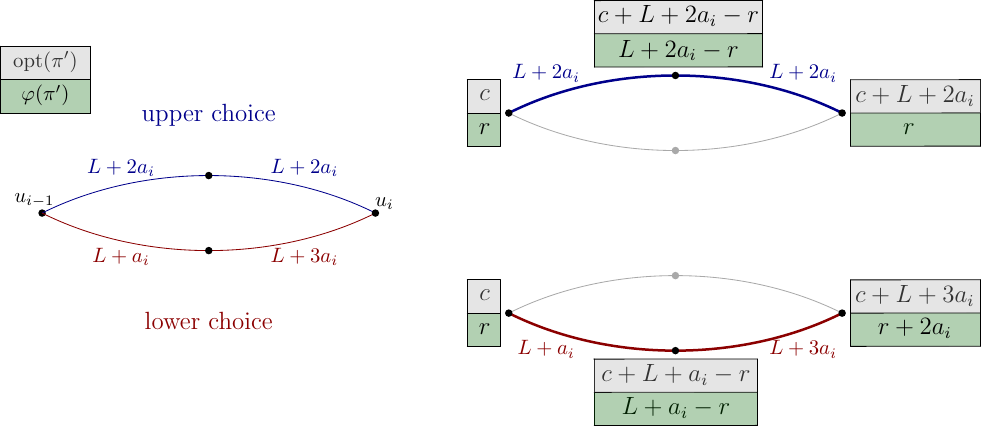}
		\caption{Left: Upper and lower choice at $i$.
			Right: the change in $\opt(\pi')$ and $\varphi(\pi')$ depending
				on whether the path is extended by the upper or the lower choice.}
		\label{fig:installation2}
	\end{figure}

	Assume that we have a path $\pi'$ that goes from $s=u_0$ to $u_{i-1}$
	with $\varphi(\pi') \le L$.
	Let $\pi'_u$ be the concatenation of $\pi'$ with the upper choice,
	and let $\pi'_\ell$ be the concatenation of $\pi'$ with the lower choice.
	Because of Lemma~\ref{le:prolong}, we obtain that
	$\opt(\pi'_u)=\opt(\pi')+L+2a_i$ and $\varphi(\pi'_u)=\varphi(\pi')$, while
	$\opt(\pi'_\ell)=\opt(\pi')+L+ 3 a_i$ and 
	$\varphi(\pi_\ell)=\varphi(\pi') +2a_i$.
	See Figure~\ref{fig:installation2}.
	Here, the assumption $\varphi(\pi') \le L$ has been important to ensure that
	in using Lemma~\ref{le:prolong} the maximum defining $\varphi(\cdot)$ is not at $0$.
	It easily follows by induction on $i$ that, for each path
	$\pi'$ from $s=u_0$ to $u_i$, we indeed have $\varphi(\pi') \le \sum_{j=1}^i 2a_i$,
	and thus the hypothesis is fulfilled for each $i\in [n]$.
	
	The intuition here is that the lower choice has a larger cost, but
	keeps more power at the extreme of the prefix path for later use.	
	See Figure~\ref{fig:installation3} for a concrete example showing the
	values $\opt(\pi')$ and $\varphi(\pi')$ for paths $\pi'$ from $s=u_0$ to $u_i$.
	It also helps understanding the idea behind the reduction.
	
	\begin{figure}[tb]
		\centering
		\includegraphics[width=\textwidth,page=3]{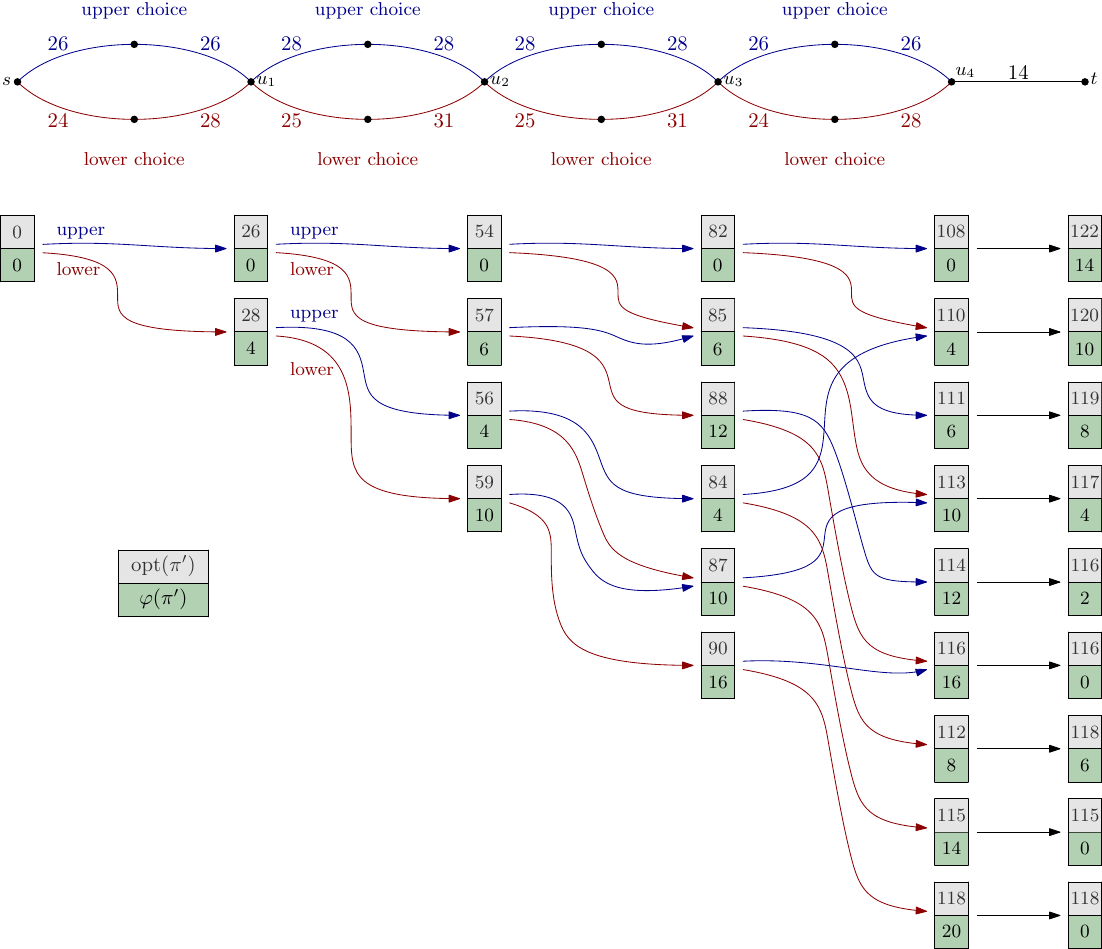}
		\caption{Top: The graph $G$ for $n=4$ with $a_1,\dots,a_4= 2,3,3,2$ and $b=7$,
			when we take $L=22$. 
			We have to decide whether there is an assignment of power with cost 
			$n L +2\sum_i a_i + b=115$ that activates some $s$-$t$ path.
			Bottom: pairs $(\opt(\pi'),\varphi(\pi'))$ 
			for all the $s$-$u_i$ paths $\pi'$.}
		\label{fig:installation3}
	\end{figure}

	Consider now a path $\pi$ from $s=u_0$ to $u_n$.
	Let $I$ be the set of indices $i\in [n]$ where the path takes the lower choice at $i$.
	From the previous discussion and a simple induction we have
	\[
		\opt(\pi) ~=~ \sum_{i\in [n]\setminus I} (L+2a_i) 
						+ \sum_{i\in I} \left(L+ 3 a_i\right)
				~=~ nL + \sum_{i\in [n]} 2 a_i + \sum_{i\in I} a_i 
	\]
	and
	\[
		\varphi(\pi) ~=~ 2\sum_{i\in I} a_i ~\le~ L.
	\]
	Since all the paths from $s$ to $u_n$ must follow the upper or lower choice at each $i\in [n]$,
	the result follows.	
\end{proof}

\begin{lemma}
\label{lem:implication}
	For any real numbers $A$ and $B$ we have
	\[
		A + \max\{ 2B-2A,0\} \le B ~~~\Longrightarrow ~~~ A=B.
	\]
\end{lemma}
\begin{proof}
	If $A\le B$, then $B-A\ge 0$ and the assumption implies $A + (2B-2A) \le B$,
	which implies $B\le A$, and thus $A=B$.
	If $A>B$, then $B-A<0$ and the assumption implies $A + 0 \le B$,
	which implies $A\le B$. Thus this cannot happen.
\end{proof}

\begin{theorem}
	The problem \textsc{Minimum Installation Path} is NP-hard.
\end{theorem}
\begin{proof}
	We show that the instance for \textsc{Subset Sum} has a positive 
	answer if and only if in the graph 
	$G=G(a_1,\dots,a_n,b)$ there is a power assignment with cost at most 
	$C:=nL+2\sum_{i\in [n]} a_i + b$
	that activates some path from $s$ to $t$.
	
	Assume that the exists a solution for the instance to the \textsc{Subset Sum} problem.
	This means that we have some $I\subseteq [n]$ such that $\sum_{i\in I}a_i =b$.
	Because of Lemma~\ref{le:feasible}, there exists a path $\pi$ from $s=u_0$ to $u_n$
	with optimal installation cost $\opt(\pi)=nL + 2 \sum_{i\in [n]} a_i + \sum_{i\in I} a_i =C$ 
	and $\varphi(\pi)=2b$. Because of Lemma~\ref{le:path},
	this means that the power assignment $p^*_\pi$
	has cost $\cost(p^*_\pi)=C$, activates all edges of $\pi$,
	and assigns power $p^*_\pi(u_n)=\varphi(\pi)=2b$ to vertex $u_n$.
	Such power assignment $p^*_\pi$ also activates the edge $u_nt$ because 
	it has weight $2b= p^*_\pi(u_n)$. (In particular, the vertex $t$ gets power $0$.)
	
	Assume now that there is a power assignment $p'\ge 0$ with cost at most $C$
	that activates a path $\pi'$ from $s$ to $t$.
	Let $\pi$ be the restriction of $\pi'$ from $s$ to $u_n$.
	Because of Lemma~\ref{le:prolong} and using that the power assignment $p'$ 
	activates $\pi'$, we have 
	\begin{equation}
		\opt(\pi) + \max \{ 2b-\varphi(\pi), 0 \} ~=~ \opt(\pi') ~\le~ \cost(p') ~\le~ C .
	\label{eq:pi}
	\end{equation}
	Because of Lemma~\ref{le:feasible}, there exists some $I\subseteq [n]$ such that
	\[
		\opt(\pi) ~=~ nL + 2\sum_{i\in [n]} a_i + \sum_{i\in I} a_i 
		~~~~~\text{and}~~~~~
		\varphi(\pi)~=~2\sum_{i\in I} a_i.
	\]
	Substituting in \eqref{eq:pi}, for such $I\subseteq [n]$ we have
	\[
		 nL + 2 \sum_{i\in [n]} a_i + \sum_{i\in I} a_i  + \max \left\{ 2b- 2\sum_{i\in I} a_i, 0 \right\}
             ~\leq~ C ~=~ nL+ 2 \sum_{i\in [n]} a_i + b.
	\]
	This means that 
	\[ 
		\sum_{i\in I} a_i  + \max \left\{ 2b-2\sum_{i\in I} a_i, 0 \right\} ~\leq~ b.
	\]
	Because of Lemma~\ref{lem:implication} we conclude that $\sum_{i\in I} a_i = b$, 
	and the given instance to \textsc{Subset Sum} problem has a solution.
\end{proof}

%%%%%%%%%%%%%%%%%%%%%%%%%%%%%%%%%%%%%%%%%%%%%%%%%%%%%%%%%%%%%%%%%%%%%%%%%%%%%%%%%%%%%%%%%%%%%%%%%%%%%%%%%%%%%%%%%%%%%%%%%%%%%%%%%%%%%%%%%%%%%%%%%%%%%%%%%%%%%%%%%%%%%%%%%%%%%%%%%%%%
\section{Minimum Barrier Shrinkage}
\label{sec:barrier}

In this section, we show that the \textsc{Minimum Barrier Shrinkage} problem
is NP-hard. The structure of the proof is very similar to the 
proof given in Section~\ref{sec:reduction} for the NP-hardness of the problem
\textsc{Minimum Installation Path}.

We first give the construction assuming that we can compute algebraic
numbers to infinite precision. Then we explain how an approximate construction
with enough precision suffices and can be computed in polynomial time.
 
\medskip

The \DEF{penetration depth} of a pair $(D(c,r),D(c',r'))$ of open disks
$D(c,r)$ and~$D(c',r')$
is $r+r'-|c-c'|$, where $|c-c'|$ is the distance between the centers $c$ and $c'$.
When no disk contains the center of the other disk, and they intersect,
then the intersection $D(c,r) \cap D(c',r')$ is a lens of width
equal to the penetration depth. See Figure~\ref{fig:barrier_penetration}.
If we shrink the disks to $D(c,r-\delta)$ and $D(c',r'-\delta')$,
the disks intersect if and only if $\delta+\delta'$ is strictly
smaller than the penetration depth. (Recall that disks are taken as open sets.)
Thus, the penetration depth equals
the minimum total shrinking of the disks so that a curve can 
pass between the two disks. 

\begin{figure}[tb]
	\centering
	\includegraphics[scale=1,page=7]{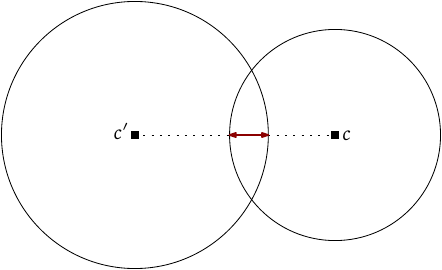}
	\caption{The penetration depth of the pair of drawn disks is the length of the arrow.}
	\label{fig:barrier_penetration}
\end{figure}

\medskip

We reduce again from \textsc{Subset Sum}.
Consider an instance $I$ of \textsc{Subset Sum} given by
a sequence $a_1,\dots,a_n$ of positive integers and a positive integer $b$.
Set $L$ to be an integer strictly larger than $2\sum_{i\in [n]} a_i$. 
Then, for each $I\subseteq [n]$ we have $2\sum_{i\in I} a_i< L$.
Set $C=nL+2\sum_{i\in [n]} a_i+b$ and $\lambda = 10 C$.
We will construct an instance to \textsc{Minimum Barrier Shrinkage} problem
such that it has a solution if and only if the instance $I$ for 
\textsc{Subset Sum} has a solution.

Figure~\ref{fig:barrier_overview} shows the overall idea of the construction.
Most of the action is happening around the filled (blue and green)
disks.
The remaining white disks create corridors to communicate from one side to the other 
of the filled disks.
To provide a feasible solution of cost at most~$C$, we have to indicate how to 
shrink the disks for a total radius of at most~$C$ and provide an $x$-$y$ curve
in the plane that does not touch the (interior of the) shrunken disks.

\begin{figure}[tb]
	\centering
	\includegraphics[scale=.9,page=1]{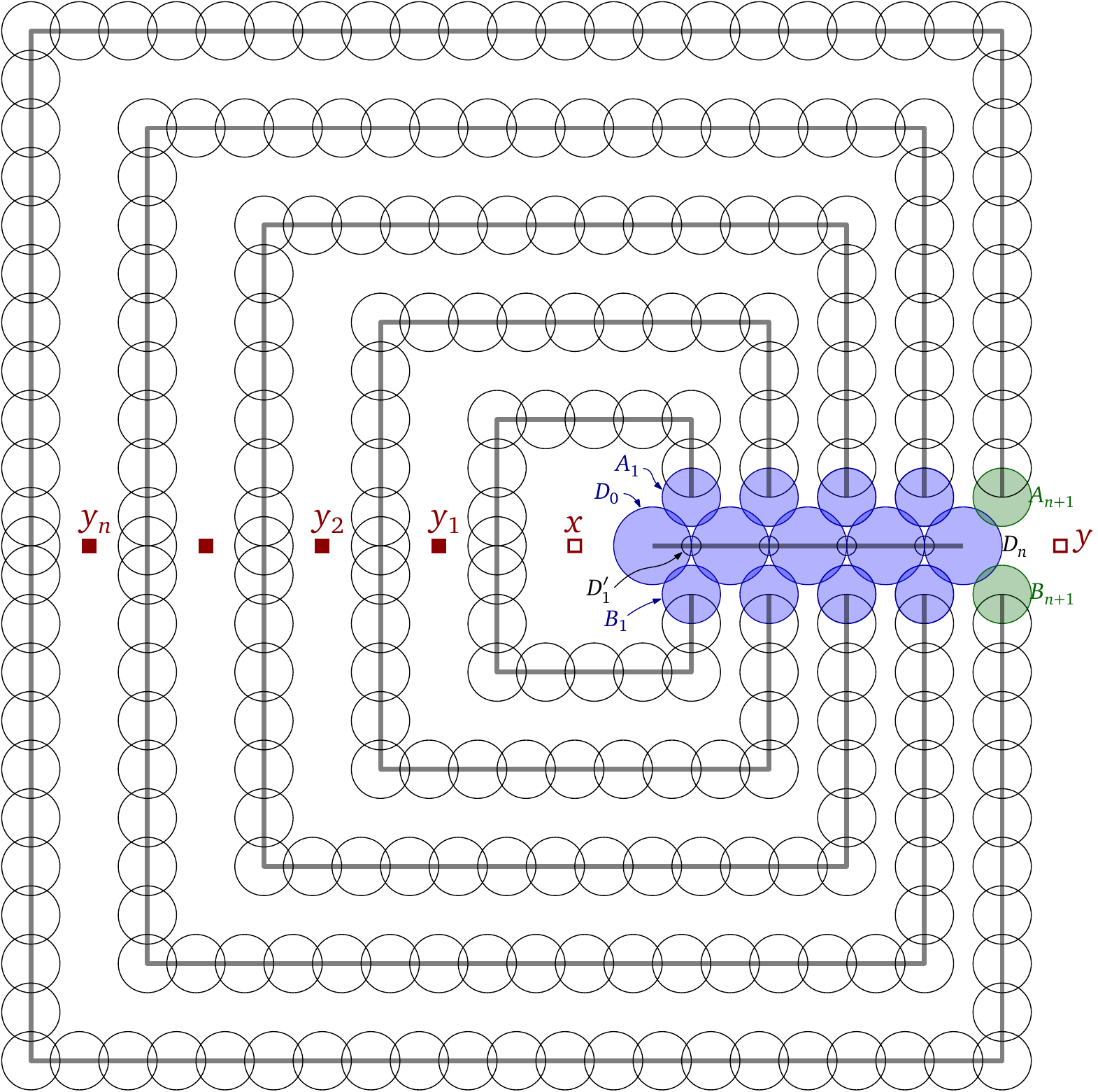}
	\caption{Basic idea of the construction for $n=4$. All the shrinking of disks 
		and the decisions on how to route the $x$-$y$ curve are happening around the 
		(blue and green) filled disks. The thick lines will not be crossed by any $x$-$y$ curve
		that is disjoint from the shrunken disks in a solution with the desired cost.}
	\label{fig:barrier_overview}
\end{figure}

In our construction, no point of the plane will be covered by more than two disks.
In such a case, the $x$-$y$ curve can be described combinatorially 
by a sequence of pairs of disks such that, for each pair $(D,D')$,
the curve passes between $D$ and $D'$, after shrinking.

If the penetration depth of two disks is at least $\lambda=10 C$, 
then, in any shrinking of the disks with total cost at most~$C$, those two disks keep intersecting,
which means that we cannot route the $x$-$y$ curve between those two disks.
More precisely, the segment connecting the centers of such disks cannot
be crossed by the $x$-$y$ curve.
In the drawings we indicate this with a thick segment connecting
the centers of the disks.

\begin{figure}[tb]
	\centering
	\includegraphics[scale=.9,page=6]{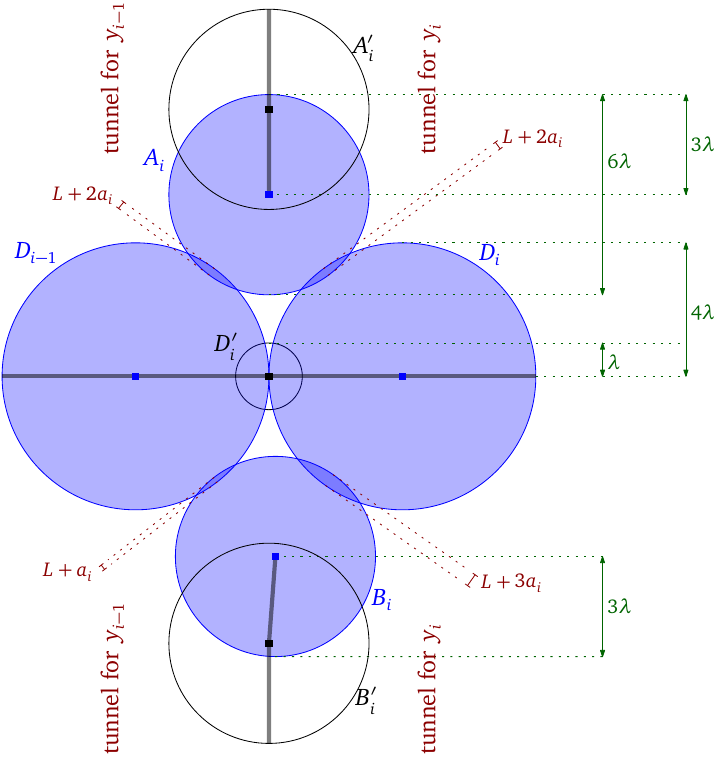}
	\caption{Block $\BB_i$ for $1<i<n$; the penetration depths are not to scale. Note that $A_i$ has the same
		overlap with $D_{i-1}$ and $D_i$, while $B_i$ is moved closer to $D_i$. 
		The center of $B_i$ is to the right of the (vertical) 
		line through the centers of $A'_i$, $A_i$, $D'_i$ and $B'_i$.}
	\label{fig:barrier_block}
\end{figure}

The main part to encode the instance, around the filled disks, consists of the following disks.
See Figures~\ref{fig:barrier_block} and~\ref{fig:barrier_end}.
\begin{itemize}
	\item For $i=0,\dots,n$, a disk $D_i$ of radius $4\lambda$ centered at $((4\lambda)\cdot 2i,0)$;
	\item for each $i\in [n]$, a disk $D'_i$ of radius $\lambda$ centered at 
		$((4\lambda)\cdot (2i-1),0)$;
	\item for each $i\in [n]$, a disk $A_i$ (for \emph{above}) of radius $3\lambda$ placed
		such that the center is above the $x$-axis, and
		the penetration depth of $(A_i,D_{i-1})$ and $(A_i,D_i)$ is $L+2a_i$;
		this means that the distance between $\Center(A_i)$ and $\Center(D_{i-1})$ 
		is $7\lambda- (L+2a_i)$, and the distance between $\Center(A_i)$ and $\Center(D_i)$ 
		is $7\lambda- (L+2a_i)$;
	\item for each $i\in [n]$, a disk $B_i$ (for \emph{below}) of radius $3\lambda$ placed
		such that the center is below the $x$-axis,  
		the penetration depth of $(B_i,D_{i-1})$ is $L+a_i$ and
		the penetration depth of $(B_i,D_i)$ is $L+3a_i$;
		this means that the distance between $\Center(B_i)$ and $\Center(D_{i-1})$ is $7\lambda- (L+a_i)$
		and the distance between $\Center(B_i)$ and $\Center(D_i)$ is $7\lambda- (L+3a_i)$;
	\item a disk $A_{n+1}$ of radius $3\lambda$ placed such that the center is above the $x$-axis,  
		the $x$-coordinate of $\Center(A_{n+1})$ is $(4\lambda)\cdot (2n+1)$,
		and the penetration depth of $(A_{n+1},D_n)$ is $2b$; $A_{n+1}$ is one of the green disks
		in the figures; 
	\item a disk $B_{n+1}$ of radius $3\lambda$ placed such that the center is below the $x$-axis,  
		the $x$-coordinate of $\Center(B_{n+1})$ is $(4\lambda)\cdot (2n+1)$,
		and the penetration depth of $(B_{n+1},D_n)$ is $2b$; $B_{n+1}$ is another of the green disks in the figures;
	\item for each $i\in [n+1]$, a disk $A'_i$ of radius $3\lambda$ 
		centered at $((4\lambda)\cdot (2i-1),8\lambda)$ and 
		a disk $B'_i$ of radius $3\lambda$ 
		centered at $((4\lambda)\cdot (2i-1),-8\lambda)$.
\end{itemize}

\begin{figure}[tb]
	\centering
	\includegraphics[scale=.8,page=8]{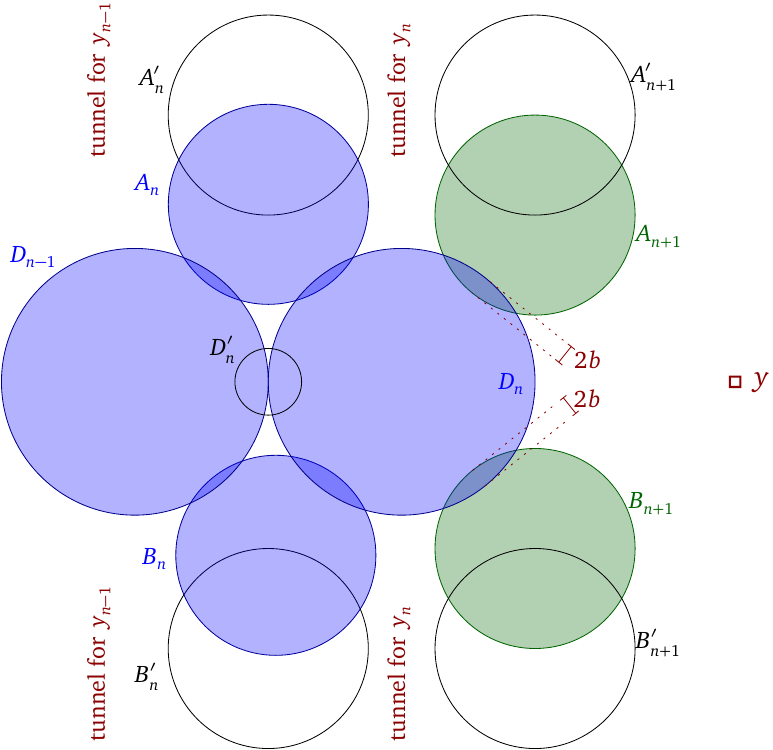}
	\caption{The blocks $\BB_n$ and $\BB_{n+1}$; the penetration depths are not to scale.}
	\label{fig:barrier_end}
\end{figure}

For $i\in [n]$, the \DEF{block} $\BB_i$ consists of the disks 
$D_{i-1}$, $D_{i}$, $D'_i$, $A_i$, $A'_i$, $B_i$ and $B'_i$.
We also define the block $\BB_{n+1}$ as the group of disks $D_n$, $A_{n+1}$, $A'_{n+1}$, $B_{n+1}$ and $B'_{n+1}$.
Note that the blocks $\BB_i$ and $\BB_{i+1}$, for $i\in [n]$, share the disk $D_i$.

For each $i\in [n+1]$, we make a \emph{path} of disks of radius $3\lambda$, starting from
$A'_i$ and finishing with $B'_i$, where any two consecutive disks have penetration
depth at least $3\lambda$. The disks in these paths are pairwise disjoint for different indices $i$,
and disjoint from the rest of the construction. The disks in each such path can be 
centered along a $5$-link axis-parallel path, and it uses $O(i)$ disks.
See Figure~\ref{fig:barrier_overview}.
We denote the path of disks for the index $i\in [n+1]$ by $\Pi_i$.
For later use, we place a point $y_i$ in the ``tunnel'' between the paths $\Pi_i$ and $\Pi_{i+1}$.
See Figure~\ref{fig:barrier_overview}.

\begin{lemma}
\label{lem:barrier1}
	For each $i\in [n]$,
	the disks $D'_i$, $A_i$ and $B_i$ are pairwise disjoint.
	Moreover, the penetration depth of the pairs $(D_{i-1},D'_i)$, $(D_i,D'_i)$,
	$(A_i,A'_i)$ and $(B_i,B'_i)$ is at least $\lambda$.
	For $\BB_{n+1}$, the disks $A_{n+1}$ and $B_{n+1}$ are disjoint and
	the penetration depth of the pairs $(A_{n+1},A'_{n+1})$ and $(B_{n+1},B'_{n+1})$ is at least $\lambda$.
\end{lemma}
\begin{proof}
	We consider only the case $i\in [n]$. The arguments for $\BB_{n+1}$ are similar.
	The penetration depth of the pairs $(D_{i-1},D'_i)$ and $(D_i,D'_i)$ is $\lambda$
	by construction.
	
	Consider the disk $\tilde A_i$ of radius $3\lambda$ centered at 
	$((4\lambda)\cdot (2i-1),5\lambda)$ and the disk $\tilde B_i$ of radius $3\lambda$
	centered at $((4\lambda)\cdot (2i-1),-5\lambda)$.
	See Figure~\ref{fig:tilde_version}. 
	We will compare $B_i$ to $\tilde B_i$; note that they have the same size,
	just a different placement. The argument for $A_i$ is the same.

	\begin{figure}[tb]
		\centering
		\includegraphics[scale=.8,page=4]{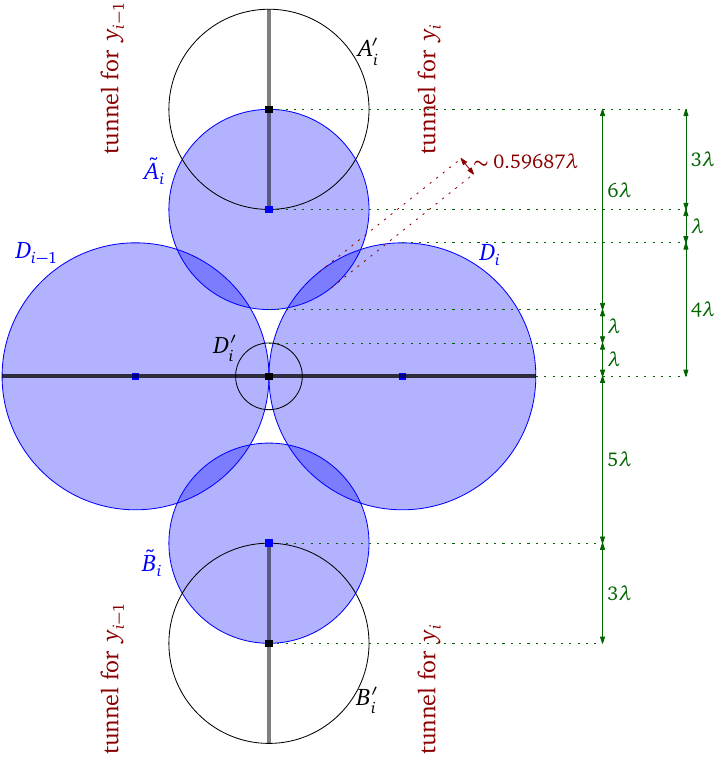}
		\caption{Disks $\tilde A_i$ and $\tilde B_i$ considered in the 
			proof of Lemma~\ref{lem:barrier1}.}
		\label{fig:tilde_version}
	\end{figure}

	The penetration depth of the pairs $(\tilde B_i,D_{i-1})$ and $(\tilde B_i,D_i)$
	is
	\[
		3\lambda+4\lambda - \sqrt{(5\lambda)^2+(4\lambda)^2} ~=~ 
		(7-\sqrt{41})\lambda ~\approx ~ 0.59687\lambda,
	\]
	while the penetration depth of $(\tilde B_i,B'_i)$ is exactly $3\lambda$.
	The disk $D'_i$ is at distance $\lambda$ from $\tilde B_i$.
	
	Since the penetration depth of $(B_i,D_{i-1})$ and $(B_i,D_i)$ is
        at most $L+3a_i \le C =\lambda/10$, these penetration depths are
        smaller than the penetration depths of $(\tilde B_i,D_{i-1})$ and
        $(\tilde B_i,D_i)$, namely, between 0 and $0.59687\lambda$.  See
        Figure~\ref{fig:tilde_version2}.  As can be seen on the figure (and
        proved by a slightly involved computation), this implies that $B_i$
        and $D'_i$ are disjoint, and that the disk $B'_i$ contains the
        center of $B_i$.  The latter fact implies that the penetration
        depth of $(B_i,B'_i)$ is at least $3\lambda$.
        \begin{figure}[tb]
        \centering
        \includegraphics[scale=.8,page=5]{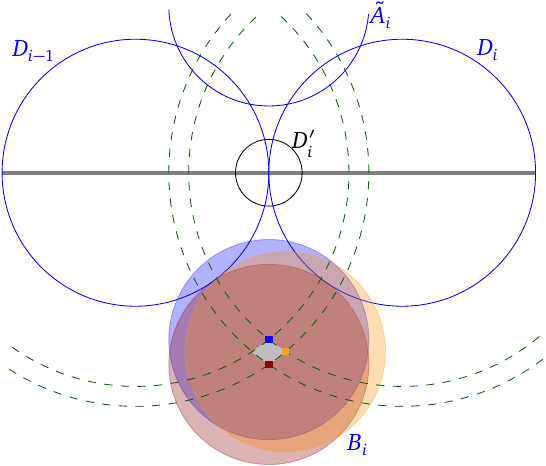}
        \caption{The gray region shows the position of the centers where the disk
        	$B_i$ may be placed, more precisely, the positions for $\Center(B_i)$ where
        	the penetration depth of $(B_i,D_{i-1})$ and $(B_i,D_i)$ lies in
        	the interval $[0, 0.59687\lambda]$.
        	The blue mark denotes the center of $\tilde B_i$.}
        \label{fig:tilde_version2}
      \end{figure}
\end{proof}

From Lemma~\ref{lem:barrier1} we conclude that, in any solution with cost under $\lambda= 10 C$,
the $x$-$y$ curve cannot cross the segments connecting $\Center(D_{i-1})$ and $\Center(D_i)$,
the segments connecting $\Center(A_i)$ and $\Center(A'_i)$, nor the segments connecting 
$\Center(B_i)$ and $\Center(B'_i)$, for each $i\in [n+1]$. Furthermore,
it cannot cross the path $\Pi_i$ connecting $A'_i$ to $B'_i$, for each $i\in [n+1]$.  
This implies that, at each block $\BB_i$, we have to decide whether the $x$-$y$ curve goes
above (crossing $A_i$ before shrinking) or below (crossing $B_i$ before shrinking). 
See Figure~\ref{fig:barrier2} for one such choice.

\begin{figure}[tb]
	\centering
	\includegraphics[scale=.9,page=2]{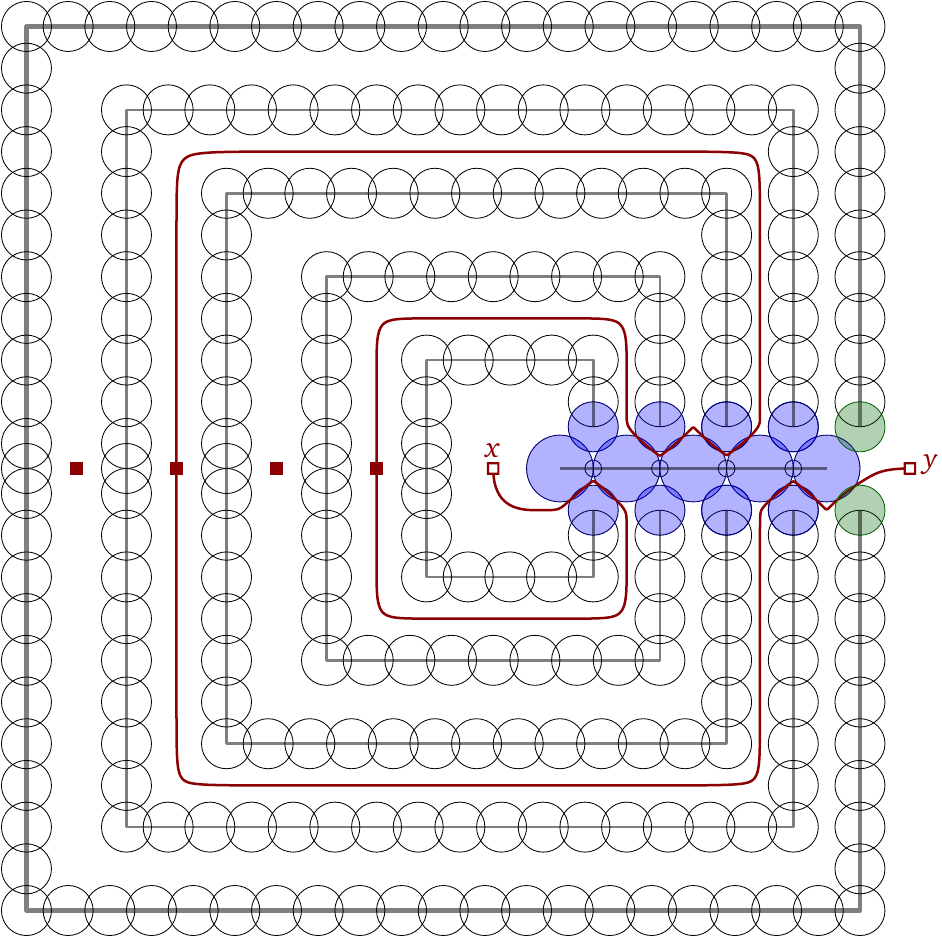}
	\caption{The red $x$-$y$ curve shows the type of decisions that have to be made
		to make a feasible solution. In this example, we have to decide 5 times
		independently whether the $x$-$y$ curve is routed above or below.
		Note that the curve can be routed to pass through each $y_i$, if desired.}
	\label{fig:barrier2}
\end{figure}

So, in a nutshell, the strategy is to reformulate the problem in terms of graphs, and to note
that the instance is equivalent to the \textsc{Minimum Installation Path} in
that graph.  Let $X_i=A_i$ or $X_i=B_i$, depending on the choice of how to route the $x$-$y$ curve.
If $X_i=A_i$, then the $x$-$y$ curve, after shrinking the disks, passes
between $D_{i-1}$ and $A_i$, and also between $D_i$ and $A_i$.
If $X_i=B_i$, then the $x$-$y$ curve, after shrinking the disks, passes
between $D_{i-1}$ and $B_i$, and also between $D_i$ and $B_i$.
Note that we can assume that the $x$-$y$ curve passes between two disks at most once.
Moreover, for each disk $D$, the $x$-$y$ curve passes between $D$ and another disk
at most twice.
Once we decide the combinatorial routing of the $x$-$y$ curve, that is, once
we select $X_1,\dots,X_n,X_{n+1}$,
then greedily shrinking the disks gives an optimal solution, similarly to Lemma~\ref{le:path}:
it pays off to push the shrinking towards disks that are crossed later by the $x$-$y$ curve.
That is, to pass between $D_1$ and $X_1$, it pays off to do not shrink $D_1$, 
as it is never crossed again, and shrink $X_1$ just enough to pass in between.
Similarly, it pays off to shrink $D_2$ to pass between $D_2$ and $X_1$, because $X_1$ will not
be crossed again later on.
In general, to pass between $D_{i-1}$ and $X_i$ it pays off to reduce $X_i$ just enough to pass between them,
taking into account how much $D_{i-1}$ was already shrunken,
and to pass between $X_i$ and $D_i$ it pays off to reduce $D_i$ just enough to pass between them,
taking into account how much $X_i$ was already reduced.

Let $\DD=\DD(I)$ be the set of all disks in the constructed instance.
\begin{lemma}
\label{le:equivalence}
	The instance $I=(a_1,\dots,a_n,b)$ to \textsc{Subset Sum} has a solution
	if and only if the instance $(\DD,x,y,C)$ to \textsc{Minimum
          Barrier Shrinkage} has a positive answer, where  $C=nL+2\sum_{i\in [n]} a_i+b$. 
Furthermore, for any $\alpha\in[0,1)$, \textsc{Minimum Barrier
          Shrinkage}$(\DD,x,y,C)$ has a positive answer if and only if \textsc{Minimum Barrier
          Shrinkage}$(\DD,x,y,C+\alpha)$ has a positive answer.
\end{lemma}
\begin{proof}
	We construct a graph $G'$ as follows.
	We make a node for each connected component of $\RR^2\setminus \bigcup \DD$ that
	may be crossed by the $x$-$y$ curve after shrinking disks for a cost strictly smaller than $\lambda=10C$.
	This means that we have the following nodes in the graph:
	\begin{itemize}
	\item a node for the cell containing $x$, which we call $x$ also; 
	\item a node for the cell containing $y$, which we call $y$ also; 
	\item a node called $\alpha_i$ for the region bounded between the disks $D_{i-1},D_i,A_i$ ($i\in [n]$);
	\item a node called $\beta_i$ for the region bounded between the disks $D_{i-1},D_i,B_i$ ($i\in [n]$);
	\item a node for the cell that contains $y_i$ ($i\in [n]$), 
		that is, the tunnel bounded by $\Pi_i$ and $\Pi_{i+1}$; we call the node $y_i$ also.
	\end{itemize}
	We put an edge between two nodes whenever we can pass from one region to the other
	passing between two disks with penetration strictly below $\lambda=10C$.
	See Figure~\ref{fig:graph} for the resulting graph, $G'$.
	This graph $G'$ is essentially the graph $G(a_1,\dots,a_n,b)$ used in Section~\ref{sec:reduction}.
	(The only difference is that, in $G'$, we have two parallel edges from $y_n$ to $y$, instead of a single edge.)

	\begin{figure}[tb]
		\centering
		\includegraphics[page=3]{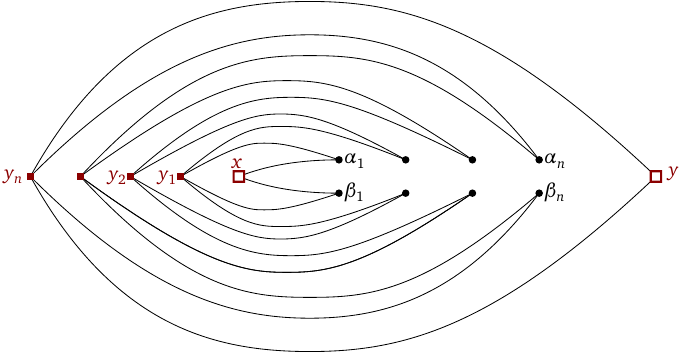}
		\caption{The combinatorially different $x$-$y$ curves can be encoded in a graph, denoted $G'$. 
			This is essentially the same graph $G(a_1,\dots,a_n,b)$ used in Section~\ref{sec:reduction},
			but with a different drawing; see Figure~\ref{fig:installation1}.}
		\label{fig:graph}
	\end{figure}

	We assign a weight to each edge of $G'$ equal to the penetration depth of the pair of
	disks that separate the cell. For example, the edges $y_{i-1 } \alpha_i$ and $\alpha_i y_i$ have
	weight $L+2a_i$ ($i\in [n]$), the edge $\beta_i y_i$ has weight $L+3a_i$ ($i\in [n]$), 
	and the two parallel edges $y_n y$ have weight $2b$.

	There is a simple correspondence between power assignments $p(\cdot)$ that 
	give a feasible solution for \textsc{Minimum Installation Path}$(G',x,y,C)$
	and the reduction in radii for feasible solutions for \textsc{Minimum Barrier Shrinkage}$(\DD,x,y,C)$,
	as follows:
	\begin{itemize}
		\item the decrease in radius of $D_i$ corresponds to the power $p(y_i)$ ($i\in [n]$);
		\item the decrease in radius of $A_i$ corresponds to the power $p(\alpha_i)$ ($i\in [n]$);
		\item the decrease in radius of $B_i$ corresponds to the power $p(\beta_i)$ ($i\in [n]$);
		\item the decrease in radius of $D_0$ corresponds to the power $p(x)$; 
		\item we may assume that at most one of the disks $A_{n+1}$ and $B_{n+1}$
			is shrunken; the decrease in radius of $A_{n+1}$ or $B_{n+1}$, whichever is larger,
			corresponds to the power $p(y)$; 
		\item we may assume that all other disks are not shrunken.
	\end{itemize}
	This correspondence transforms feasible solutions for \textsc{Minimum Installation Path}$(G',x,y,C)$
	into feasible solutions for \textsc{Minimum Barrier
          Shrinkage}$(\DD,x,y,C)$, and conversely.  So the instances
	\textsc{Minimum Installation Path}$(G',x,y,C)$ and \textsc{Minimum
          Barrier Shrinkage}$(\DD,x,y,C)$ are equivalent.

	The second part of the lemma follows from the above correspondence
        and from Lemma~\ref{le:integral}.
\end{proof}

The disks $\DD$, as described, cannot be constructed in polynomial time in a Turing machine
because the centers of the disks do not have integer (or rational) coordinates.
More precisely, the centers of $A_i$ and $B_i$ ($i\in [n+1]$) are solutions to a system of equations
with degree-two polynomials. However, we can scale up the numbers involved in
the construction, and then round the non-integer numbers, to obtain
a polynomial-time construction, doable in a Turing machine:

\begin{theorem}
The \textsc{Minimum Barrier Shrinkage} problem is NP-hard.
\end{theorem}
\begin{proof}
  Consider an instance $(a_1,\dots,a_n,b)$ for \textsc{Subset Sum} and the
  associated instance $(\DD,x,y,C)$ for \textsc{Minimum Barrier Shrinkage}
  constructed above, with $C=nL+2\sum_{i\in [n]} a_i+b$.

  The centers of the disks in
  $\DD\setminus \{ A_1,\dots,A_{n+1},B_1,\dots, B_{n+1} \}$ are integers
  bounded by $O(\lambda)=O(nL)$.  For each $i\in [n+1]$, we compute the
  centers of the disks $A_i$ and $B_i$ up to a precision of at least
  $\eps= \frac{1}{6(n+1)}$. Thus, the coordinates of the centers are
  multiples of $\eps$.  Let $\hat A_i$ and $\hat B_i$ be the resulting
  disks; they have the same radius, $3\lambda$, but have been displaced by
  at most $\eps$ with respect to the original position in the construction.
  Let $\hat \DD$ be the set of disks obtained from $\DD$, where each
  $A_i,B_i$ are replaced with $\hat A_i,\hat B_i$ ($i\in [n+1]$).

  We consider instances of \textsc{Minimum Barrier Shrinkage}.  If the
  instance $(\DD,x,y,C)$ is positive, then the instance
  $(\hat\DD,x,y,C+1/3)$ is also positive (because each of the $2(n+1)$
  disks are moved by at most $\eps$, so the total displacement is at
  most~$1/3$), which implies that the instance $(\DD,x,y,C+2/3)$ is also
  positive (by the same argument), which in turn also implies that the
  instance $(\DD,x,y,C)$ is positive (by Lemma~\ref{le:equivalence}).  So,
  the instances $(\DD,x,y,C)$ and $(\hat\DD,x,y,C+1/3)$ are equivalent.
  
  Scaling all values in the construction of~$\hat\DD$ (coordinates and
  radii) by $1/\eps$, we get a construction where the disks have centers
  with integer coordinates, the radii are integers, and the whole
  construction can be constructed in polynomial time.
\end{proof}

Note that it is not clear whether the \textsc{Minimum Barrier Shrinkage} problem
belongs to NP. Indeed, if some triples of disks intersect,
a priori it seems that a solution may have to reduce the radius of
some disks by non-rational numbers, and decisions at different parts depend on each other,
which could increase the algebraic degree of the numbers telling how much to decrease the radii.

\paragraph{Remark}
A similar statement can be done for axis-parallel squares. For this we have to place the overlapping
squares in such a way that the overlap region, an axis-parallel rectangle, has width equal to the value 
we want to encode ($L+a_i$, $L+2a_i$, $L+3a_i$ or $2b$).
In such a case we do not run into the numerical issues with the centers because all the coordinates
can be taken directly to be integers.

\section*{Acknowledgments}

We would like to thank the reviewers for their careful reading of our
manuscript, their suggestions, and pointing out the possibility
to get the FPTAS discussed in Section~\ref{sec:installation_fptas}.  
This research was partially
supported by the Slovenian Research Agency (P1-0297, J1-8130, J1-8155).
Part of the research was done while Sergio was invited professor at
Universit\'e Paris-Est.

%%%%%%%%%%%%%%%%%%%%%%%%%%%%%%%%%%%%%%%%%%%%%%%%%%%%%%%%%%%%%%%%%%%%%%%%%
% Bibliography
%%%%%%%%%%%%%%%%%%%%%%%%%%%%%%%%%%%%%%%%%%%%%%%%%%%%%%%%%%%%%%%%%%%%%%%%%


\begin{thebibliography}{10}

\bibitem{alqahtani}
Hasna~Mohsen Alqahtani and Thomas Erlebach.
\newblock Minimum activation cost node-disjoint paths in graphs with bounded
  treewidth.
\newblock In Viliam Geffert, Bart Preneel, Branislav Rovan, Julius Stuller, and
  A~Min Tjoa, editors, {\em {SOFSEM} 2014: Theory and Practice of Computer
  Science - 40th International Conference on Current Trends in Theory and
  Practice of Computer Science, Nov{\'{y}} Smokovec, Slovakia, January 26-29,
  2014, Proceedings}, volume 8327 of {\em Lecture Notes in Computer Science},
  pages 65--76. Springer, 2014.

\bibitem{AltCGK17}
Helmut Alt, Sergio Cabello, Panos Giannopoulos, and Christian Knauer.
\newblock Minimum cell connection in line segment arrangements.
\newblock {\em Int. J. Comput. Geometry Appl.}, 27(3):159--176, 2017.

\bibitem{althaus}
Ernst Althaus, Gruia C{\u{a}}linescu, Ion~I. Mandoiu, Sushil~K. Prasad,
  N.~Tchervenski, and Alexander Zelikovsky.
\newblock Power efficient range assignment for symmetric connectivity in static
  ad hoc wireless networks.
\newblock {\em Wireless Networks}, 12(3):287--299, 2006.

\bibitem{BeregK09}
Sergey Bereg and David~G. Kirkpatrick.
\newblock Approximating barrier resilience in wireless sensor networks.
\newblock In Shlomi Dolev, editor, {\em Algorithms for Sensor Systems,
  {ALGOSENSORS} 2009}, volume 5804 of {\em Lecture Notes in Computer Science},
  pages 29--40. Springer, 2009.

\bibitem{cabello2018minimum}
Sergio Cabello, Kshitij Jain, Anna Lubiw, and Debajyoti Mondal.
\newblock Minimum shared-power edge cut.
\newblock {\em Networks}, 75(3):321--333, 2020.

\bibitem{ChanK12}
David Yu~Cheng Chan and David~G. Kirkpatrick.
\newblock Approximating barrier resilience for arrangements of non-identical
  disk sensors.
\newblock In Amotz Bar{-}Noy and Magn{\'{u}}s~M. Halld{\'{o}}rsson, editors,
  {\em Algorithms for Sensor Systems, {ALGOSENSORS} 2012}, volume 7718 of {\em
  Lecture Notes in Computer Science}, pages 42--53. Springer, 2013.

\bibitem{GareyJ79}
Michael~R. Garey and David~S. Johnson.
\newblock {\em Computers and Intractability: A Guide to the Theory of
  NP-Completeness}.
\newblock W. H. Freeman \& Co., 1979.

\bibitem{KormanLSS13}
Matias Korman, Maarten Löffler, Rodrigo~I. Silveira, and Darren Strash.
\newblock On the complexity of barrier resilience for fat regions and bounded
  ply.
\newblock {\em Computational Geometry}, 72:34--51, 2018.

\bibitem{Kumar:2005:BCW:1080829.1080859}
Santosh Kumar, Ten~H. Lai, and Anish Arora.
\newblock Barrier coverage with wireless sensors.
\newblock In {\em Proceedings of the 11th Annual International Conference on
  Mobile Computing and Networking (MobiCom)}, pages 284--298. ACM, 2005.

\bibitem{kumar2007barrier}
Santosh Kumar, Ten~H. Lai, and Anish Arora.
\newblock Barrier coverage with wireless sensors.
\newblock {\em Wireless Networks}, 13(6):817--834, 2007.

\bibitem{lando-nutov}
Yuval Lando and Zeev Nutov.
\newblock On minimum power connectivity problems.
\newblock In Lars Arge, Michael Hoffmann, and Emo Welzl, editors, {\em
  Algorithms - {ESA} 2007, 15th Annual European Symposium, Eilat, Israel,
  October 8-10, 2007, Proceedings}, volume 4698 of {\em Lecture Notes in
  Computer Science}, pages 87--98. Springer, 2007.

\bibitem{nutovactivation}
Zeev Nutov.
\newblock Activation network design problems.
\newblock In Teofilo~F. Gonzalez, editor, {\em Handbook of Approximation
  Algorithms and Metaheuristics, Second Edition}, chapter~15. CRC Press, 2018.

\bibitem{Panigrahi2011}
Debmalya Panigrahi.
\newblock Survivable network design problems in wireless networks.
\newblock In {\em Proceedings of the Twenty-second Annual ACM-SIAM Symposium on
  Discrete Algorithms (SODA)}, pages 1014--1027, 2011.

\bibitem{TsengK11}
Kuan-Chieh~Robert Tseng and David Kirkpatrick.
\newblock On barrier resilience of sensor networks.
\newblock In Thomas Erlebach, Sotiris Nikoletseas, and Pekka Orponen, editors,
  {\em Algorithms for Sensor Systems, {ALGOSENSORS} 2011}, volume 7111 of {\em
  Lecture Notes in Computer Science}, pages 130--144. Springer, 2012.

\end{thebibliography}
\end{document}